\documentclass[11pt]{article}

\usepackage[T1]{fontenc}
\usepackage[utf8]{inputenc}
\usepackage{lmodern}
\usepackage[a4paper]{geometry}

\usepackage{amsmath,amsthm,amssymb}
\newtheorem{theorem}{Theorem}
\newtheorem{corollary}[theorem]{Corollary}
\newtheorem{lemma}[theorem]{Lemma}

\newtheorem{definition}[theorem]{Definition}

\usepackage{algorithm}
\usepackage{algpseudocode}

\usepackage{bm}
\usepackage{xcolor}
\usepackage{fullpage}
\usepackage{xspace}
\usepackage{paralist}

\usepackage{ifpdf}
\newcommand{\mydriver}{hypertex}
\ifpdf
\renewcommand{\mydriver}{pdftex}
\fi
\usepackage[breaklinks,\mydriver]{hyperref}
\usepackage{cleveref}

\newcommand{\congest}{\textsf{CONGEST}\xspace}
\newcommand{\congc}{\textsf{CONGESTED CLIQUE}\xspace}
\newcommand{\local}{\textsf{LOCAL}\xspace}
\newcommand{\MPC}{\textsf{MPC}\xspace}

\newcommand{\cj}{\ensuremath{c}\xspace}
\newcommand{\Exp}[1]{\mathbf{E}\left[#1\right]}
\newcommand{\Prob}[1]{\mathbf{Pr}\left[#1\right]}	
\newcommand{\nat}{\ensuremath{\mathbb{N}}}

\def\epsilon{\ensuremath{\varepsilon}}
\newcommand{\eps}{\ensuremath{\epsilon}}
\newcommand{\machines}{\ensuremath{M}}
\renewcommand{\machines}{\ensuremath{\mathfrak{M}}}
\newcommand{\spac}{\ensuremath{S}}
\renewcommand{\spac}{\ensuremath{\mathfrak{s}}}
\newcommand{\Gin}{\ensuremath{\mathfrak{G}}}
\newcommand{\nin}{\ensuremath{\mathfrak{n}}}
\newcommand{\nim}{\ensuremath{\mathfrak{m}}}

\newcommand{\qual}{\ensuremath{\mathfrak{q}}}

\newcommand{\COMMENTED}[1]{{}}
\newcommand{\junk}[1]{\COMMENTED{#1}}
\newcommand{\hide}[1]{{}}

\title{\textbf{Simple, Deterministic, Constant-Round Coloring in the Congested Clique}}
\author{\textbf{Artur Czumaj}
	\thanks{Department of Computer Science and Centre for Discrete Mathematics and its Applications (DIMAP), University of Warwick. Email: A.Czumaj@warwick.ac.uk. Research partially supported by the Centre for Discrete Mathematics and its Applications (DIMAP), by a Weizmann-UK Making Connections Grant, by IBM Faculty Award, and by EPSRC award EP/N011163/1.}
	\and
	\textbf{Peter Davies}
	\thanks{Institute of Science and Technology Austria (IST Austria). Email: peter.davies@ist.ac.at.  Research partially supported by the European Union's Horizon 2020 research and innovation programme under the Marie Sk\l odowska-Curie grant agreement No. 754411.}
	\and
	\textbf{Merav Parter}
	\thanks{Weizmann Institute of Science, Rehovot, Israel. Email: merav.parter@weizmann.ac.il. Research partially supported by a Weizmann-UK Making Connections Grant.}
}
\date{}

\begin{document}
\maketitle
\begin{abstract}
We settle the complexity of the $(\Delta+1)$-coloring and $(\Delta+1)$-list coloring problems in the \congc model by presenting a simple \emph{deterministic} algorithm for both problems running in a constant number of rounds. This matches the complexity of the recent breakthrough \emph{randomized} constant-round $(\Delta+1)$-list coloring algorithm due to Chang et al.\ (PODC'19), and significantly improves upon the state-of-the-art $O(\log \Delta)$-round deterministic $(\Delta+1)$-coloring bound of Parter\ (ICALP'18).

A remarkable property of our algorithm is its simplicity. Whereas the state-of-the-art \emph{randomized} algorithms for this problem are based on the quite involved local coloring algorithm of Chang et al.\ (STOC'18), our algorithm can be described in just a few lines. At a high level, it applies a careful derandomization of a recursive procedure which partitions the nodes and their respective palettes into separate bins. We show that after $O(1)$ recursion steps, the remaining uncolored subgraph within each bin has linear size, and thus can be solved locally by collecting it to a single node. This algorithm can also be implemented in the Massively Parallel Computation (\MPC) model provided that each machine has linear (in $\nin$, the number of nodes in the input graph) space.

We also show an extension of our algorithm to the \MPC regime in which machines have \emph{sublinear} space: we present the first deterministic $(\Delta+1)$-list coloring algorithm designed for sublinear-space \MPC, which runs in $O(\log \Delta + \log\log \nin)$ rounds.

\end{abstract}

\section{Introduction}
In this paper, we study deterministic complexity of vertex coloring in two fundamental models of distributed and parallel computation: the \congc and the Massively Parallel Computation models.

The \emph{\congc} model is a classic and prominent model of distributing computing, introduced by Lotker et al.\ \cite{LPP15}. It is a variant of the standard \congest message-passing model that allows all-to-all communication, where in every single round, each node can send $O(\log n)$ bits to each other node.
The \emph{Massively Parallel Computation (\MPC)} model, introduced by Karloff et al.\ \cite{KSV10}, is a nowadays standard theoretical model for parallel algorithms. It shares many similarities to earlier models of parallel computation, but it also allows for unlimited local computation, making it close to various models of distributed computing, e.g., to \congc.

The $(\Delta+1)$-coloring problem and its variations are considered to be corner-stone problems of local computation and are arguably among the most fundamental graph problems in parallel and distributed computing with numerous implications. In the \emph{$(\Delta+1)$-coloring problem}, where $\Delta$ is the maximum degree of the input graph $\Gin = (V,E)$, the goal is to color each node in $\Gin$ with a color in $\{1, \ldots, \Delta+1\}$ such that no two neighbors have the same color. We also consider its generalization, the \emph{$(\Delta+1)$-list coloring problem}, where each node has a (possibly different) palette of $\Delta+1$ colors, and the goal is to compute a legal coloring in which each node is assigned a color from its own palette. Finally, the most restrictive variant of the problem is the \emph{$(\deg+1)$-list coloring problem}, which is the same as the $(\Delta+1)$-list coloring problem except that the size of each node $u$'s palette is $\deg(u)+1$, which might be considerably smaller than $\Delta+1$. It is easy to see that in each of the variations of the $(\Delta+1)$ coloring problem above there is always a legal coloring, and the main task is to understand the complexity of finding the sought coloring.

Our main result is that we can solve the $(\Delta+1)$-list coloring problem in a constant number of rounds of \congc, thereby settling the deterministic complexity of $(\Delta+1)$-coloring and $(\Delta+1)$-list coloring in this model.

\begin{theorem}
	\label{thm:CC-coloring}
	\textbf{Deterministic $(\Delta+1)$-list coloring} can be performed in $\mathbf{O(1)}$ rounds in \congc.
\end{theorem}

We obtain this by presenting the following linear-space \MPC result, and then applying known reductions between the \congc and \MPC models (cf. Section \ref{subsec:relation-between-models}):

\begin{theorem}
	\label{thm:lMPC-coloring}
	\textbf{Deterministic $(\Delta+1)$-list coloring} can be performed in $\mathbf{O(1)}$ rounds in \MPC with $O(\nin)$ local space per machine and $O(\nin\Delta)$ total space.
\end{theorem}

We note that for general $(\Delta+1)$-list coloring, the input size is $\Theta(\nin\Delta)$ and therefore the global \MPC space bound in Theorem \ref{thm:lMPC-coloring} is optimal. However, for the special case of $(\Delta+1)$-coloring, where initial palettes are all $\{1, \ldots, \Delta+1\}$ (and so need not be specified as input), the input size is $\Theta(\nim+\nin)$. We remark that we can also obtain optimal global space in this case.

\begin{theorem}
	\label{thm:MPC-coloring-congc}
	\textbf{Deterministic $(\Delta+1)$-coloring} can be performed in $\mathbf{O(1)}$ rounds in \MPC with $O(\nin)$ local space per machine and $O(\nim+\nin)$ total space.
\end{theorem}

Finally, we extend our approach to the low-space \MPC regime, and provide the first $(\deg+1)$-list coloring (and therefore $(\Delta+1)$-list coloring and $(\Delta+1)$-coloring) algorithm designed for that model.

\begin{theorem}
	\label{thm:MPC-coloring}
	For any constant $\eps>0$, \textbf{deterministic $(\deg+1)$-list coloring} can be performed in \MPC in $O(\log\Delta + \log\log\nin)$ rounds, using $O(\nin^{\eps})$ space per machine and $O(\nim+\nin^{1+\eps})$ total space.
\end{theorem}

\subsection{Models of computation: \congc and \MPC}

In the \congc model of distributed computation, our input is a graph $\Gin$ with $\nin$ nodes and $\nim$ edges, and the aim is to solve some graph problem by performing computation at the \emph{nodes of the input graph}. Each nodes' initial input consists of its adjacent edges; we then proceed in synchronous \emph{rounds}, in which each node can perform some local computation, and then send $O(\log \nin)$-bit messages to every other node (i.e., unlike in the related \local and \congest models, communication is \emph{not} restricted to the edges of the input graph). We require all nodes to eventually output their local part of the problem's solution (in the case of coloring problems, this means the color they have chosen), and our aim is to minimize the number of communication \emph{rounds} required to do so.

The Massively Parallel Computation (\MPC) model is a parallel computing model first introduced by Karloff et al. \cite{KSV10}. It recently gained a significant popularity as it provides a clean abstraction of a number of massively parallel computation frameworks extensively used in applications, such as MapReduce~\cite{DG04,DG08}, Hadoop~\cite{White12}, Dryad~\cite{IBYBF07}, or Spark~\cite{ZCFSS10}. In the \MPC model, there are $\machines$ machines, each with $\spac$ available space. The input graph $\Gin$ is initially distributed arbitrarily across the machines (i.e., we need at least enough machines so that the total space $\machines\cdot\spac$ is at least the input size). We progress again in synchronous rounds, in which each machine can perform some local computation, and then send a message to each other machine. Unlike in \congc, we do not bound the size of each message; we instead require that the total information sent and received by each machine in each round must fit into the machine's local space (i.e., is of size at most $\spac$). This in particular implies that the total communication of the \MPC model is bounded by $\machines\cdot\spac$ in each round. At the end of the computation, machines must collectively output the solution; any machine can output any part so long as the total output is a complete and consistent solution. Our aim is again to minimize the number of communication rounds; we also now have local space $\spac$ and total space $\machines\cdot\spac$ as model parameters. Space regimes of particular interest are the linear-space regime (where $\spac = \Theta(\nin)$), and the low-space regime ($\spac = \Theta(\nin^\epsilon)$ for any constant $\epsilon \in (0,1)$). Here we measure space in terms of $O(\log \nin)$-bit words, so that $O(\nin)$ space is sufficient to store e.g. a node's neighborhood.
\subsection{Relationship between models}
\label{subsec:relation-between-models}

The two models we study here have their provenance in different fields: \congc from distributed computing, and \MPC from parallel computing. Because of this, \congc is stated as a message-passing model, with the computation performed by \emph{nodes of the input graph}, whose initial input consists of the adjacent edges to that node (but communication is \emph{not} restricted to the input graph, and can be between any pair of nodes). In \MPC, on the other hand, the computation is performed by \emph{machines} that are not associated with any particular part of the input graph, and indeed the input is initially distributed arbitrarily across machines.

It has been noted, though, that (under certain conditions on algorithms) \congc is equivalent to \MPC using $O(\nin)$ space per machine and $O(\nin^2)$ total space \cite{BDH18}. In particular, due to the constant round routing algorithm of Lenzen \cite{Lenzen13}, information can be redistributed essentially arbitrarily, so there is no need to associate the computational entities with nodes in the input graph. This is in stark contrast to the related \local and \congest distributed models, in which the link between computation and input graph locality is integral.

Since we give results both for \congc and low-space \MPC, we will adopt the \MPC perspective, which can accommodate both. That is, we will consider our computation to be done on \emph{machines} with \spac~ space, and give results for the linear-space regime ($\spac = \Theta(\nin)$) which directly apply also to \congc, and then for the low-space regime ($\spac = \Theta(\nin^\epsilon)$).

\subsection{Related work}
Coloring in \congc is a problem which has received considerable recent attention, with a succession of results on randomized algorithms:
\subsubsection{State of the art for randomized coloring algorithms.}


The starting point for the current $(\Delta+1)$ coloring randomized \congc algorithms is the following: provided that the maximum degree $\Delta=O(\sqrt{\nin})$, nodes can efficiently simulate the \local $(\Delta+1)$-list coloring algorithm of Chang et al.\ (CLP) within $O(\log^* \Delta)$ rounds. Building upon this observation, Parter \cite{Parter18} showed that the degree of the nodes can be reduced by applying $O(\log\log\Delta)$ steps of vertex partitioning. At each recursive partitioning step, the current graph was split into $poly(\Delta)$ vertex-disjoint subgraphs plus one additional left-over subgraph. These vertex-disjoint subgraphs were colored recursively in parallel, and the left-over subgraph was colored only once the coloring of the other subgraphs in its recursion level was complete. This approach led to an $O(\log\log\Delta \cdot \log^*\Delta)$-round algorithm. Parter and Su \cite{ParterSu18} improved this complexity to $O(\log^* \Delta)$, by modifying the internal details of the CLP algorithm to allow it to be simulated when $\Delta=\nin^{1/2+o(1)}$. In a recent breakthrough, Chang et al.\ \cite{chang2019complexity} presented a $O(1)$ round procedure which was obtained by
\begin{inparaenum}[(i)]
	\item modifying the recursive partitioning of \cite{Parter18}, and
	\item modifying aspects of the CLP algorithm.
\end{inparaenum}
The key innovative part in their recursive partitioning was to partition both the vertices and the \emph{colors} in their respectability into independent coloring instances. This new partitioning idea was useful to support the more general $(\Delta+1)$-list coloring problem, as well as to provide efficient implementation low-space \MPC, where (upon incorporating the subsequent polylogarithmic-round network decomposition algorithm of Rozho{\v{n}} and Ghaffari \cite{rozhon2019polylogarithmic}), the algorithm of Chang et al. \ \cite{chang2019complexity} works in $O(\log\log\log \nin)$ rounds. This round complexity is matched by a conditional lower bound due to Ghaffari et al. \cite{GKU19}, which states that an $o(\log\log\log \nin)$-round randomized component-stable coloring algorithm in low-space \MPC would imply a $2^{\log^{o(1)}\log \nin}$-round deterministic \local coloring algorithm, and a $\log^{o(1)}\log \nin$-round randomized one. (Note, though, that our algorithms are \emph{not} component-stable, since they involve global agreement on hash functions.)

\subsubsection{State of the art for deterministic coloring algorithms.}

Censor-Hillel et al. \cite{CPS17} presented a quite general scheme for derandomization in the \congc model by combining the methods of bounded independence with efficient computation of the conditional expectation.  Via a reduction to the maximal independent set (MIS) problem, they gave a $(\Delta+1)$-coloring algorithm that runs in $O(\log\Delta)$ rounds provided that $\Delta = O(\nin^{1/3})$. Parter \cite{Parter18} showed a deterministic $(\Delta+1)$-coloring in $O(\log \Delta)$ rounds, that works for any value of $\Delta$. Barenboim and Khazanov \cite{barenboim2018distributed} presented improved deterministic local algorithms as a function of the graph's \emph{arboricity}. Finally, concurrently with our work, Bamberger et al.\ \cite{BKM20} give $(\deg+1)$-list coloring algorithms requiring $O(\log\Delta\log\log\Delta)$ rounds in \congc, $O(\log^2 \Delta)$ in linear-space \MPC, and $O(\log^2 \Delta+\log \nin)$ in low-space \MPC. These bounds are significantly weaker than ours, but we note that the low-space \MPC algorithm of \cite{BKM20} has the advantage of using optimal $O(\nim+\nin)$ global space, compared to our $O(\nim + \nin^{1+\eps})$, and that for \congc and linear-space \MPC we present algorithms only for $(\Delta+1)$-list coloring, rather than $(\deg+1)$-list coloring.

Prior to this, we are aware of no existing deterministic low-space \MPC algorithms for the $(\Delta+1)$-list coloring problem, though for the special case of $(\Delta+1)$-coloring one can simulate the PRAM algorithm of Han \cite{Han96} in $\tilde O(\log^2\nin)$ rounds.

\subsection{Our approach}

We present a simple, deterministic, constant-round \congc algorithm for $(\Delta+1)$-list coloring, by carefully derandomizing a recursive partitioning procedure similar to those used by Parter \cite{Parter18} and Chang et al.\ \cite{chang2019complexity}.

The first step in our approach is to show a randomized procedure which maps nodes of the input graph and colors from nodes' palettes into bins. We show that if we leave only one of the bins without colors, most nodes in the remaining bins receive more colors from their palette into their bin than neighbors. This means that the graphs induced by the nodes in each bin can be colored recursively; we must also show that we can separately color the nodes in the bin which did not receive colors, and the nodes which did not receive more colors than neighbors (which can happen with small probability).

Then, we show that this randomized binning procedure requires only bounded independence. This means a \emph{small random seed} suffices to provide all of the necessary randomness. We can then apply a distributed implementation of the \emph{method of conditional expectations} in order to \emph{deterministically} select such a seed that performs well.

The final step is to show that after each recursive application of this procedure, we significantly reduce the size of the graph induced by the nodes in each bin. In fact, after only \emph{nine} recursive calls, we reduce the size of each instance from $O(\nin\Delta)$ to $O(\nin)$, at which point it is small enough to simply collect onto a single machine and solve locally.

In low-space \MPC only this final step must be altered; since now machines have only $O(\nin^\epsilon)$ space, we cannot solve instances locally and instead apply a reduction to maximal independent set (MIS) in order to use an existing  $O(\log\Delta + \log\log \nin)$-round MIS algorithm of Czumaj et al.\ \cite{CDP19}.

\section{Preliminaries}
In this section we detail some existing fundamentals on which our work relies.

\subsection{Communication in \congc and \MPC}
\label{subsec:communication-in-fully-scalable-MPC}

As noted in \cite{CDP19}, prior work on MapReduce, a precursor to \MPC, facilitates efficient deterministic communication in \MPC, so long as the information sent and received by each machine adheres to the relevant space bounds of the model. In particular, the following results from \cite{GSZ11} provide all the communication primitives we require (and use of features of MapReduce which are common to \MPC):

\begin{lemma}[\cite{GSZ11}, as stated in \cite{CDP19}]
	\label{lem:comm}
	For any positive constant $\delta$, sorting and computing prefix sums of $\nin$ numbers can be performed deterministically in MapReduce (and therefore in \MPC and \congc) in $O(1)$ rounds using $\spac = \nin^\delta$ space per machine and $O(\nin)$ total space.
\end{lemma}

The computation of prefix sums here means the following: each machine $m \in [\machines]$ holds an input value $x_m$, and outputs $\sum_{i=1}^{m}x_i$.

This result essentially allows us to perform all of the communication we will need to do in a constant number of rounds. For example, if for each edge $\{u,v\}$ we create two entries $(u,v)$ and $(v,u)$ in memory, and the sort these lexicographically, we can ensure that the neighborhoods of all nodes are stored on contiguous blocks of machines. Then, by computing prefix sums, we can compute sums of values among a node's neighborhood, or indeed over the whole graph. This allows us to, e.g., compute objective functions for the method of conditional expectations (see Section \ref{sec:condexp}). Where 2-hop neighborhoods fit in the space of a single machine, we can collect them by sorting edges to collect $1$-hop neighborhoods onto machines, and then having each such machine send requests for the neighborhoods of all the nodes it stores.

An important point to note is that since Lemma \ref{lem:comm} uses $\spac = n^\delta$ for \emph{any} positive constant $\delta$, by setting $\delta$ sufficiently smaller than our space parameter $\epsilon$, we can perform $n^{\Omega(1)}$ \emph{simultaneous} sorting or prefix sum procedures. This will be especially useful to us for efficiently performing the method of conditional expectations.

The \congc model contains the additional restriction of an $O(\log \nin)$-bit bound on the information any node can send to any other in a round (as opposed to \MPC, where only the \emph{total} amount of information a node sends, to all others, is bounded). This restriction is circumvented by Lenzen's routing algorithm \cite{Lenzen13}, which in $O(1)$ rounds allows all information to be routed to the correct destinations so long as all nodes obey an $O(\nin)$ bound on the \emph{total} amount of messages they send and receive.

So, henceforth we need not concern ourselves with routing information to the correct machines; so long as all machines obey the total space bounds at all times, we can do so in a generic fashion using these existing results.

\subsection{Bounded-independence random variables}

To convert randomized algorithms into deterministic algorithms, we follow a classic derandomization recipe: we show that a randomized algorithm can (or be modified to) work using only \emph{bounded independence}, which means that it requires only a short random \emph{seed}, and then we apply the \emph{method of conditional expectations} to efficiently search the space of random seeds and deterministically fix one that performs well.

The first step of this process, showing that a randomized algorithm works using only bounded-independence variables, will make heavy use of the following \emph{concentration bound} for sums of such variables.

\begin{lemma}[Lemma 2.2 of \cite{BR94}]
	\label{lem:conc}
	Let $\cj \ge 4$ be an even integer. Suppose $Z_1, \dots, Z_t$ are $\cj$-wise independent random variables taking values in $[0,1]$. Let $Z = Z_1 + \dots + Z_t$ and $\mu = \Exp{Z}$, and let $\lambda > 0$. Then,
	\begin{align*}
	\Prob{|Z - \mu| \ge \lambda} &\le 2 \left(\frac{\cj t}{\lambda^2}\right)^{\cj/2}
	\enspace.
	\end{align*}
\end{lemma}

Here, $\cj$-wise independence means that any $\cj$ of the variables $Z_1, \dots, Z_t$ can be considered to be independent, but groups of \emph{more} than $\cj$ cannot.

\subsection{Families of bounded-independence hash functions}

To generate $\cj$-wise independent random variables from a short random seed, we use the concept of families of bounded-independence hash functions:

\begin{definition}
	For $N, L, \cj \in \nat$ such that $\cj \le N$, a family of functions $\mathcal{H} = \{h : [N] \rightarrow [L]\}$ is \emph{$\cj$-wise independent} if for all distinct $x_1, \dots, x_\cj \in [N]$, the random variables $h(x_1), \dots, h(x_\cj)$ are independent and uniformly distributed in $[L]$ when $h$ is chosen uniformly at random from~$\mathcal{H}$.
\end{definition}

That is, the outputs of a uniformly random function from a $\cj$-wise independent family are $\cj$-wise independent random variables. Such families of small size (and therefore requiring a small seed to specify any particular element) are known to exist; we will use the following well-known lemma (see, e.g., Corollary~3.34 in \cite{Vadhan12}):

\begin{lemma}
	\label{lem:hash}
	For every $a$, $b$, $\cj$, there is a family of $\cj$-wise independent hash functions $\mathcal{H} = \{h : \{0,1\}^a \rightarrow \{0,1\}^b\}$ such that choosing a random function from $\mathcal{H}$ takes $\cj \cdot \max\{a,b\}$ random bits, and evaluating a function from $\mathcal{H}$ takes $poly(a,b,\cj)$ computation.
\end{lemma}

In particular, we will be using families $\mathcal{H}$ of hash functions with $\nin^{O(1)}$ range and domain, and some large constant $\cj$, so by Lemma \ref{lem:hash}, a uniformly random function from $\mathcal{H}$ can be specified using $O(\log \nin)$ random bits (and this string of bits is what we call its \emph{seed}). This also means that functions can be evaluated using polylogarithmic computation, and therefore space, which is possible locally even in low-space \MPC.

The domain and range of the hash functions must be powers of $2$, and therefore we incur some error when our desired range is not (for domain we simply ignore any excess). If, for example, we have some stated range $r$, we will actually use range $\{0,1\}^{\lceil\log  (r \nin^3)\rceil}$, and then map intervals of this range as equally as possible (i.e., differing in size by at most $1$) to $[r]$. We then have error in our resulting probabilities of $O(\nin^{-3})$, which for our purposes will be clearly negligible. Note, too, that outputs of our `simulated' hash function with range $r$ are still exactly $c$-wise independent (and therefore we can still use Lemma \ref{lem:conc}), but are no longer exactly uniform.

\subsection{Method of conditional expectations}
\label{sec:condexp}

To deterministically choose a good hash function from our small families, we employ the classical \emph{method of conditional expectations} (see, e.g., \cite[Chapter~3.4]{Vadhan12} or \cite[Section~2.4]{CDP19}). In our context, we will show that, over the choice of a uniformly random hash function $h \in \mathcal{H}$, the expectation of some cost function $q$ (which is a sum of functions $q_x$ calculable by individual machines) is at most some value $Q$. That is,
\begin{align*}
\mathbf{E}_{h \in \mathcal{H}} \left[{q(h):=\sum_{\text{machines } x} q_x(h)}\right]
&\le
Q
\enspace.
\end{align*}

The probabilistic method implies the existence of a hash function $h^* \in \mathcal{H}$ for which $q(h^*) \le Q$, and our goal is to find one such $h^* \in \mathcal{H}$ in $O(1)$ rounds of \congc or low-space \MPC.

We will find the sought hash function $h^*$ by deterministically fixing the $O(\log \nin)$-bit seed that specifies it (cf. Lemma \ref{lem:hash}). To do so, we have all machines agree gradually on chunks of $\delta\log\nin$ bits at a time, for sufficiently small constant $\delta$. That is, we iteratively extend a fixed, agreed, prefix $h_{pre}$ of the final hash function $h^*$, $\delta\log\nin$ bits at a time, until we have fixed the entire seed.

For each chunk, and for each $i$, $1 \le i \le \nin^\delta$, we write $i$ as a $\delta\log\nin$-length bit string. Then, each machine $x$ calculates $\mathbf{E}_{h \in \mathcal{H}} \left[ q_x(h) | \Xi_{h_{pre}+i}\right]$, where $\Xi_{h_{pre}+i}$ is the event that the seed specifying $h$ is prefixed by the current fixed prefix $h_{pre}$, followed by $i$. We then sum these values over all machines for each $i$ concurrently (using Lemma \ref{lem:comm} and choosing $\delta$ sufficiently smaller than the space parameter $\epsilon$), obtaining $\mathbf E_{h \in \mathcal{H}} \left[ q(h) | \Xi_{h_{pre}+i} \right]$. By the probabilistic method, at least one of these values is at most $\mathbf E_{h \in \mathcal{H}} \left[ q(h) | \Xi_{h_{pre}} \right]$. We fix a value of $i$ such that this is the case, append this to $h_{pre}$, and continue.

After a constant number of iterations, we find the entire string of bits to define a hash function $h^* \in \mathcal{H}$. Since in every iteration we did not increase the conditional expectation $\mathbf E_{h \in \mathcal{H}} \left[ q(h) | \Xi_{h_{pre}} \right]$, we have
\[ q(h^*) \le \mathbf E_{h \in \mathcal{H}} \left[ q(h)\right]\le Q\enspace.\]
Each iteration requires only a constant number of rounds in \congc or \MPC, so this process takes only a constant number of rounds in total.

\section{Coloring in \congc and Linear-Space \MPC}
We are now ready to present our constant-round algorithm for $(\Delta+1)$-list coloring in \congc. Our algorithm will be a recursive procedure, which we call \textsc{ColorReduce}. To color an input graph $\Gin$ we will call \textsc{ColorReduce}$(\Gin,\Delta)$ (and in subsequent calls, $\ell$ will serve as an approximation of maximum degree). We will denote by $\nin$ the number of nodes in $\Gin$, and when analyzing a subsequent recursive \textsc{ColorReduce}$(G,\ell)$ we denote by $n_G$ the number of nodes in $G$.

\begin{algorithm}[H]
	\caption{\textsc{ColorReduce}$(G,\ell)$}
	\label{alg:ColorReduce}
	\begin{algorithmic}
		\State If $G$ has size $O(\nin)$: collect $G$ onto a single machine and color $G$ locally.
		\State Otherwise: $G_0, \dots, G_{\ell^{0.1}} \gets \textsc{Partition}(G,\ell)$.
		\State Let $\ell' = \ell^{0.9} - \ell^{0.6}$.
		\State For each $i = 1, \dots, \ell^{0.1}-1$, perform \textsc{ColorReduce}$(G_i,\ell')$ in parallel.
		\State Update color palettes of $G_{\ell^{0.1}}$, perform \textsc{ColorReduce}$(G_{\ell^{0.1}},\ell')$.
		\State Update color palettes of $G_0$, collect $G_0$ onto a single machine and color locally.
	\end{algorithmic}
\end{algorithm}

The algorithm relies on a procedure for partitioning a coloring instance into multiple instances of lower size and lower maximum degree. Specifically, we partition the node set into $\ell^{0.1}$ bins, and then partition the colors into all but the final bin (and restrict nodes' palettes to the colors sent to their bin). Next, we show that (after possibly removing some \emph{bad nodes}, to be defined shortly) each node (not in the final bin) has more colors from its palette assigned to its bin than neighbors. This means that we can recursively color the graphs induced by these bins in parallel, since the color palettes of any nodes from different bins are now completely disjoint. Then, we must afterwards color the final bin $\ell$ which was not previously assigned colors, and the graph of bad nodes. For that, in the last two steps in \textsc{ColorReduce}, before we recursively color $G_{\ell^{0.1}}$ and then color $G_0$, we first update color palettes of all nodes in $G_{\ell^{0.1}}$ and $G_0$, respectively, by removing from the palettes of every node $v$ the colors used by all its neighbors (in $G_0, \dots, G_{\ell^{0.1}}$) already colored before.

\begin{algorithm}[H]
	\caption{\textsc{Partition}$(G,\ell)$}
	\label{alg:Partition}
	\begin{algorithmic}
		\State Let hash function $h_1: [\nin] \to [\ell^{0.1}]$ map nodes $v$ to a bin $h_1(v) \in [\ell^{0.1}]$.
		\State Let hash function $h_2: [\nin^2] \to [\ell^{0.1}-1]$ map colors $\gamma$ to a bin $h_2(\gamma) \in [\ell^{0.1}-1]$.
		\State Let $G_0$ be the graph induced by \emph{bad} nodes.
		\State Let $G_1, \dots, G_{\ell^{0.1}}$ be the graphs induced by \emph{good} nodes in bins $1, \dots, \ell^{0.1}$ respectively.
		\State Restrict palettes of nodes in $G_1, \dots, G_{\ell^{0.1}-1}$ to colors assigned by $h_2$ to corresponding bins.
		\State Return $G_0, \dots, G_{\ell^{0.1}}$.
	\end{algorithmic}
\end{algorithm}

Since we consider list coloring, the total number of unique colors can be up to $\nin^2$, which necessitates the larger domain for $h_2$.

To fully specify the algorithm, we must define what we mean by \emph{good} and \emph{bad nodes} and by \emph{good} and \emph{bad bins}. Let $d(v)$ denote degree of $v$ in $G$, and $d'(v)$ denote degree of $v$ within its new bin $h_1(v)$, that is, in the subgraph of $G$ induced by all nodes $u \in V$ with $h_1(u) = h_1(v)$. Let $p(v)$ be the size of the color palette of node $v$, and $p'(v)$ be the palette size of $v$ after the call: for nodes mapped to bins in $[\ell^{0.1}-1]$, $p'(v)$ denotes the number of colors from the color palette of $v$ that are hashed to the bin $h_1(v)$; for nodes mapped to bin $\ell^{0.1}$, $p'(v)$ denotes the number of colors remaining from the color palette of $v$ \emph{after} $v$'s neighbors \emph{not} in bin $\ell^{0.1}$ (that is, in $G_1, \dots, G_{\ell^{0.1}-1}$) have been colored by \textsc{ColorReduce}.

\begin{definition}[\textbf{Good and bad nodes and bins}]\mbox{}\label{def:good-and-bad}
	\begin{itemize}
		\item A node $v$ in bins $[\ell^{0.1}-1]$ is \emph{good} if $|d'(v) - d(v) \ell^{-0.1}| \le \ell^{0.6}$ and $p'(v) \ge p(v) \ell^{-0.1} + \ell^{0.7}$.
		\item A node $v$ in bin $\ell^{0.1}$ is \emph{good} if $|d'(v) - d(v) \ell^{-0.1}| \le \ell^{0.6}$.
		\item A bin in $[\ell^{0.1}]$ is \emph{good} if it contains fewer than $2n_G \ell^{-0.1}+\nin^{0.6}$ nodes.
		\item Nodes and bins are \emph{bad} if they are not \emph{good}.
	\end{itemize}
\end{definition}

Intuitively, bad nodes and bins are those whose behavior differs significantly from what we would expect in a random partition of $G$ into $\ell^{0.1}$ node-disjoint graphs, with random $h_1$ and $h_2$.

\paragraph{Overview of the analysis of \textsc{ColorReduce}.}

In order to analyze the complexity of \textsc{ColorReduce}, we have to understand how to implement each single call to \textsc{Partition}, and how often we will be making the recursive calls. Further, in order the algorithm to work correctly, we have to ensure that every time we collect $G_0$ (containing bad nodes) onto a single machine in order to color it, the total size of $G_0$ is $O(\nin)$ (since otherwise $G_0$ would not fit a single machine).

In Section \ref{subsec:maintaining-invariant}, we describe a useful invariant for \textsc{Partition} call that is maintained throughout the entire run of \textsc{ColorReduce}. In Section \ref{subsec:bad-and-good-for-random-hash-function}, we estimate the number of bad nodes and bins for random hash functions $h_1$ and $h_2$ from appropriate families of hash functions. In Section \ref{subsec:choice-of-hash-functions}, we describe how to deterministically choose hash functions $h_1$ and $h_2$ (from the same families of hash functions) so that these estimations are still ensured, and that we guarantee $G_0$ is sufficiently small.
Next, in Section \ref{sec:recursion}, we estimate the number of recursive calls performed in the call to \textsc{ColorReduce}$(G,\Delta)$ and subsequent recursive calls. Our main result (cf. Lemma \ref{lem:9-calls}) is that after a depth of 9 recursive calls, the graph induced by each bin is of size $O(\nin)$. This implies that the call to \textsc{ColorReduce}$(\Gin,\Delta)$ creates a recursion tree of depth at most $9$, resulting in a constant number of groups of instances of size $O(\nin)$ to be collected onto machines and colored sequentially. With these claims, in Section \ref{sec:final-analysis} we finalize our analysis of Theorem \ref{thm:CC-coloring}, and show that \textsc{ColorReduce}$(G,\Delta)$ deterministically returns $(\Delta+1)$-list coloring in $O(1)$ rounds in \congc and in $O(1)$ rounds on \MPC with $O(\nin)$ local space on any machine, and (optimal for $(\Delta+1)$-list coloring) $O(\nin\Delta)$ total space. Then, in Section \ref{subsec:reducing-space-for-Delta+1-coloring-on-MPC}, we reduce the global \MPC space for the $(\Delta+1)$-coloring problem to the optimal bound of $O(\nim+\nin)$.

\subsection{Invariant maintained during all calls to {\sc Partition}}
\label{subsec:maintaining-invariant}

We begin our analysis with presenting a central invariant of our \textsc{Partition} call that is maintained throughout the entire run of \textsc{ColorReduce} (recall that $\ell' = \ell^{0.9} - \ell^{0.6}$).


\begin{lemma}
	\label{lemma:properties-of-Partition}
	Assume that at the start of our \textsc{Partition} call, the following holds for all nodes $v$ in $G$:
	\begin{inparaenum}[(i)]
		\item $\ell < p(v)$,
		\item $d(v) \le \ell + \ell^{0.7}$,
		\item $d(v) < p(v)$.
	\end{inparaenum}
	Then for \emph{all good} nodes $v$ in $G$,
	\begin{enumerate}[(i)]
		\item $\ell' < p'(v)$,
		\item $d'(v) \le \ell' + \ell'^{0.7}$, and
		\item $d'(v) < p'(v)$.
	\end{enumerate}
\end{lemma}

Note that since we assume $d(v) \le \ell + \ell^{0.7}$, we can also assume $\ell$ is at least a sufficiently large constant, since otherwise maximum degree is also bounded by a constant, so the instance is of total size $O(\nin)$ and would have already been collected and colored on a single machine.

\begin{proof}
	Let us first consider any good node $v$ with $h_1(v) \in [\ell^{0.1}-1]$. Then, $d'(v) \le d(v) \ell^{-0.1} + \ell^{0.6}$ and $p'(v) \ge p(v) \ell^{-0.1} + \ell^{0.7}$, by Definition \ref{def:good-and-bad}. Therefore,
	\begin{enumerate}[(i)]
		\item since $\ell < p(v)$, we have $p'(v) \ge p(v) \ell^{-0.1} + \ell^{0.7} > p(v) \ell^{-0.1} > \ell^{0.9} > \ell'$;
		\item since $d(v) \le \ell + \ell^{0.7}$, we have $d'(v) \le d(v) \ell^{-0.1} + \ell^{0.6} \le (\ell^{0.9} + \ell^{0.6}) + \ell^{0.6} < \ell' + \ell'^{0.7}$;
		\item since $d(v) < p(v)$, we have $p'(v) \ge p(v) \ell^{-0.1} + \ell^{0.7} > d(v) \ell^{-0.1} + \ell^{0.6} \ge d'(v)$.
	\end{enumerate}
	
	Next, let us consider any good node $v$ with $h_1(v) = \ell^{0.1}$. Then, by Definition \ref{def:good-and-bad}, it holds $|d'(v) - d(v) \ell^{-0.1}| \le \ell^{0.6}$. Further, notice that we update the color palette of $v$ (and hence set $p'(v)$) only after coloring in \textsc{ColorReduce} all $v$’s neighbors in $G_1, \dots, G_{\ell^{0.1}-1}$. Therefore,
	\begin{enumerate}[(i)]
		\item since at most one color is removed from $v$'s palette for each neighbor of $v$ \emph{not} in bin $\ell^{0.1}$, since $|d'(v) - d(v) \ell^{-0.1}| \le \ell^{0.6}$, $d(v) < p(v)$, and $\ell < p(v)$, we have,
		\begin{align*}
		p'(v)
		&\ge
		p(v) - (d(v) - d'(v))
		\ge
		p(v) - (d(v) - (d(v) \ell^{-0.1} + \ell^{0.6}))
		=
		p(v) - d(v) (1 - \ell^{-0.1}) - \ell^{0.6}
		\\
		&>
		p(v) - p(v) (1 - \ell^{-0.1}) - \ell^{0.6}
		=
		p(v) \ell^{-0.1} - \ell^{0.6}
		>
		\ell^{0.9} - \ell^{0.6}
		=
		\ell'
		\enspace;
		\end{align*}
		\item as before: since $d(v) \le \ell + \ell^{0.7}$, we have $d'(v) \le d(v) \ell^{-0.1} + \ell^{0.6} \le (\ell^{0.9} + \ell^{0.6}) + \ell^{0.6} < \ell' + \ell'^{0.7}$;
		\item $d'(v) < p'(v)$ follows since $d(v) < p(v)$, and colors are only removed from $v$'s palette when used by neighbors.\qedhere
	\end{enumerate}
\end{proof}

Notice that at the very beginning, with our input graph $G$ of maximum degree $\Delta$ and color palettes satisfying $d(v) < p(v)$ for all $v$, our initial setting $\ell = \Delta$ ensures that before the first call of \textsc{Partition} all conditions of Lemma~\ref{lemma:properties-of-Partition} (for all nodes $v$:
\begin{inparaenum}[(i)]
	\item $\ell < p(v)$,
	\item $d(v) \le \ell + \ell^{0.7}$,
	\item $d(v) < p(v)$)
\end{inparaenum}
are satisfied, and hence, since \textsc{Partition} is only called on good nodes, Lemma~\ref{lemma:properties-of-Partition} ensures that the properties in the invariant are preserved for future calls of \textsc{Partition}.

\begin{corollary}
	\label{corollary:properties-of-Partition}
	The input instance to any call to \textsc{Partition} satisfies the following three conditions for all nodes $v$:
	\begin{inparaenum}[(i)]
		\item $\ell < p(v)$,
		\item $d(v) \le \ell + \ell^{0.7}$,
		\item $d(v) < p(v)$.
	\end{inparaenum}
\end{corollary}

\subsection{Good and bad nodes and bins for \emph{random} hash functions}
\label{subsec:bad-and-good-for-random-hash-function}

In this section we analyze the likelihood of having many bad nodes and bins in a single call to \textsc{Partition}, assuming that $h_1$ and $h_2$ are hash functions \emph{chosen at random} from the families of $c$-wise independent hash functions. (In Section \ref{subsec:choice-of-hash-functions} we will extend our analysis and show how to deterministically choose appropriate $h_1$ and $h_2$.)

Let $\mathcal{H}_1$ be a $c$-wise independent family of hash functions $h: [\nin] \to [\ell^{0.1}]$, and let $\mathcal{H}_2$ be a $c$-wise independent family of hash functions $h: [\nin^2] \to [\ell^{0.1}-1]$, for sufficiently large constant $c$.

In order to understand basic properties of pairs of hash functions from $\mathcal{H}_1$ and $\mathcal{H}_2$, we define the \emph{cost} function $\qual(h_1,h_2)$ of a pair of hash functions $h_1 \in \mathcal{H}_1$ and $h_2 \in \mathcal{H}_2$, as follows:
\begin{align}
\label{def:quality}
\qual(h_1,h_2)
&=
|\{\text{bad nodes under }h_1,h_2\}| + \nin \cdot |\{\text{bad bins under }h_1\}|
\enspace.
\end{align}

We will analyze the likelihood of having many bad nodes and bins, assuming that $h_1$ and $h_2$ are chosen at random from the families of $c$-wise independent hash functions $\mathcal{H}_1$ and $\mathcal{H}_2$, and relying on Corollary \ref{corollary:properties-of-Partition} that at the start of our \textsc{Partition} for all nodes $v$:
\begin{inparaenum}[(i)]
	\item $\ell < p(v)$,
	\item $d(v) \le \ell + \ell^{0.7}$,
	\item $d(v) < p(v)$.
\end{inparaenum}

\paragraph{Good and bad bins.}

We will first prove that with high probability, no bins are bad.

\begin{lemma}\label{lem:badbin}
	With probability at least $1-\nin^{-2}$, all bins are good.
\end{lemma}

\begin{proof}
	We apply Lemma \ref{lem:conc} with $Z_1, \dots, Z_{n_G}$ as indicator random variables for the events that each node is hashed to a particular bin $b \in [\ell^{0.1}]$. These variables are $(c-1)$-wise independent and each have expectation $\ell^{-0.1}$. Since we assume $c$ to be a sufficiently high constant, by Lemma \ref{lem:conc}, we obtain $\Prob{|Z-\Exp{Z}| \ge n^{0.6}} \le n^{-3}$.
	\junk{
		\begin{align*}
		\Prob{|Z-\mu| \ge n^{0.6}}
		&\le
		2 \left(\frac{c n_G}{n^{1.2}}\right)^{c/2}
		\le
		2 \left(\frac{c n}{n^{1.2}}\right)^{c/2}
		\le
		2 (c n^{-0.2})^{c/2}
		\le
		n^{-3}
		\enspace.
		\end{align*}
	}
	By the union bound over bins $b$, with probability at least $1 - \nin^{-2}$ all bins contain fewer than $n_G \ell^{-0.1} + \nin^{0.6}$ nodes, and are therefore good.
\end{proof}

\paragraph{Good and bad nodes.}

In order to estimate the number of bad nodes, we begin with a proof that the probability that nodes' degrees differ too much from what would be expected (from a fully random partition) is low.

\begin{lemma}\label{lem:badnode1}
	The probability that a node $v$ satisfies $|d'(v) - d(v) \ell^{-0.1}| \ge \ell^{0.6}$ is at most $\ell^{-3}$.
\end{lemma}

\begin{proof}
	We apply Lemma \ref{lem:conc} with $Z_1, \dots, Z_{d(v)}$ as indicator random variables for the events that each of $v$'s neighbors is hashed to the same bin as $v$. These variables are $(c-1)$-wise independent and each have expectation $\ell^{-0.1}$. Since we assume $c$ to be a sufficiently high constant and (cf. Corollary \ref{corollary:properties-of-Partition}) have $d(v) \le \ell + \ell^{0.7}$, by Lemma \ref{lem:conc}, we obtain $\Prob{|Z-\Exp{Z}| \ge \ell^{0.6}} \le \ell^{-3}$.
	\junk{
		\begin{align*}
		\Prob{|Z-\mu| \ge \ell^{0.6}}
		&\le
		2 \left(\frac{c d(v)}{\ell^{1.2}}\right)^{c/2}
		\le
		2 \left(\frac{c (\ell + \ell^{0.7})}{\ell^{1.2}}\right)^{c/2}
		\le
		2 (2 c \ell^{-0.2})^{c/2}
		\le
		\ell^{-3}
		\enspace.
		\end{align*}
	}
	Therefore, since $\Exp{Z} = d(v) \cdot \ell^{-0.1}$ and $Z = d'(v)$, we can conclude the required claim. 
\end{proof}

We now prove that the probability that nodes in bins $1, \dots, \ell^{0.1}-1$ do not receive enough colors to color themselves is low.

\begin{lemma}\label{lem:badnode2}
	The probability that a node $v$ with $h_1(v) \ne \ell^{0.1}$ satisfies $p'(v) \le p(v) \ell^{-0.1} + \ell^{0.7}$ is at most $\ell^{-3}$.
\end{lemma}

\begin{proof}
	We apply Lemma \ref{lem:conc}, with $Z_1,\dots,Z_{p(v)}$ as indicator random variables for the events that each color in $v$'s palette is hashed to its bin. These variables are $(c-1)$-wise independent and each have expectation $\frac{1}{\ell^{0.1}-1}$. Since we assume $c$ to be a sufficiently high constant, by Lemma \ref{lem:conc}, we obtain $\Prob{|Z-\Exp{Z}| \ge p(v)^{0.6}} \le p(v)^{-3}$.
	\junk{
		\begin{align*}
		\Prob{|Z-\mu| \ge p(v)^{0.6}}
		&\le
		2 \left(\frac{c p(v)}{p(v)^{1.2}}\right)^{c/2}
		=
		2 (2 c p(v)^{-0.2})^{c/2}
		\le
		p(v)^{-3}
		\enspace.
		\end{align*}
	}
	Therefore, since $Z = p'(v)$ and $\Exp{Z} = \frac{p(v)}{\ell^{0.1}-1}$, the probability that $p'(v) \le \frac{p(v)}{\ell^{0.1}-1} - p(v)^{0.6}$ is at most $p(v)^{-3}$. Furthermore, since $\ell < p(v)$ (cf. Corollary \ref{corollary:properties-of-Partition}), we get,
	\begin{align*}
	\frac{p(v)}{\ell^{0.1}-1} - \frac{p(v)}{\ell^{0.1}}
	&=
	\frac{p(v)}{\ell^{0.1}(\ell^{0.1}-1)}
	\ge
	\frac{p(v)}{\ell^{0.2}}
	\ge
	\frac12 \left(\frac{p(v)}{p(v)^{0.2}} + \frac{\ell}{\ell^{0.2}}\right)
	>
	p(v)^{0.6} + \ell^{0.7}
	\enspace.
	\end{align*}
	Hence, with probability at least $1 - p(v)^{-3} \ge 1 - \ell^{-3}$ we obtain the following,
	\begin{align*}
	p'(v)
	&>
	\frac{p(v)}{\ell^{0.1} - 1} - p(v)^{0.6}
	>
	\frac{p(v)}{\ell^{0.1}} + \ell^{0.7}
	\enspace.
	\qedhere
	\end{align*}
\end{proof}

We can now combine Lemmas \ref{lem:badnode1} and \ref{lem:badnode2} to bound the probability that nodes are bad:

\begin{lemma}\label{lem:badnode3}
	A node is bad with probability at most $2\ell^{-3}$.
\end{lemma}

\begin{proof}
	For a node $v$ to be bad, at least one of the bad events described by Lemmas \ref{lem:badnode1} and \ref{lem:badnode2} must occur. By the union bound, the probability of this is at most $2\ell^{-3}$.
\end{proof}

This provides a bound on the cost of random hash function pairs:

\begin{lemma}\label{lemma:expected-quality}
	The expected cost of a random hash function pair from $\mathcal{H}_1 \times \mathcal{H}_2$ is at most $\frac{\nin}{\ell^2}$.
\end{lemma}

\begin{proof}
	Combining Lemmas \ref{lem:badbin} and \ref{lem:badnode3} gives the following:
	\begin{align*}
	\Exp{\qual(h_1,h_2)} &=
	\Exp{|\{\text{bad nodes}\}| + \nin \cdot |\{\text{bad bins}\}|}
	\le
	2 \nin \ell^{-3} + \nin \cdot \nin^{-2}
	\le
	\frac{\nin}{\ell^2}
	\enspace.\qedhere
	\end{align*}
\end{proof}

\subsection{Choosing appropriate hash functions}
\label{subsec:choice-of-hash-functions}

In this section, we discuss how to deterministically choose hash functions $h_1 \in \mathcal{H}_1$ and $h_2 \in \mathcal{H}_2$ to ensure the desired properties.
First, let us recall that by Lemma \ref{lem:hash}, we can assume that families of hash functions $\mathcal{H}_1$ and $\mathcal{H}_2$ have size $\nin^{O(1)}$. Next, by Lemma \ref{lemma:expected-quality}, upon choosing a random pair of hash functions $h_1 \in \mathcal{H}_1$ and $h_2 \in \mathcal{H}_2$, we have $\Exp{\qual(h_1,h_2)} \le \frac{\nin}{\ell^2}$. Therefore, by the method of conditional expectations (cf. Section \ref{sec:condexp}), we can find a pair of hash functions $h_1 \in \mathcal{H}_1$ and $h_2 \in \mathcal{H}_2$ with cost $\qual(h_1,h_2) \le \frac{\nin}{\ell^2}$ in $O(1)$ rounds. By (\ref{def:quality}), this ensures that there are no bad bins and that there are at most $\frac{\nin}{\ell^2}$ bad nodes.

In order to implement the method of conditional expectations to find such an $h_1 \in \mathcal{H}_1$ and $h_2 \in \mathcal{H}_2$, we require that $\qual(h_1,h_2)$ is the sum of functions computable by individual machines. To achieve this, we will distribute information between the machines as follows: each node $v$ will be assigned a machine $x_v$, which will store all of its adjacent edges and its palette. The machine will then be able to determine whether a particular pair of hash functions $h_1,h_2$ results in the node being \emph{bad}, since this depends only on $v$'s palette, neighbor IDs, and the hash functions. The local cost function $q_{x_v}(h_1,h_2)$ will then simply be an indicator variable for the event that $v$ is \emph{bad} under $h_1$ and $h_2$.

We will also assign a special machine $x^*$ to check if any bins are bad (contain at least $2n_G \ell^{-0.1} + \nin^{0.6}$ nodes) for any hash function $h_1$, for which it needs only to store the IDs of each node. The local cost function $q_{x^*}(h_1,h_2)$ is the number of bad bins under $h_1$, multiplied by $\nin$, which is computable by $x^*$.

The global cost function $\qual(h_1,h_2)=|\{\text{bad nodes under }h_1,h_2\}| + \nin \cdot |\{\text{bad bins under }h_1\}|$ is then the sum of these local cost functions computable by individual machines. If we in fact use machines for multiple nodes where space allows, then (even among all concurrent instances) we will require only $O(1+\frac{\nim}{\nin})$ machines, i.e., total space of $O(\nim+\nin)$.

\begin{lemma}\label{lem:deterministic-hashing}
	In $O(1)$ \MPC rounds with $O(\nin)$ local space per machine and $O(\nim+\nin)$ total space, one can select hash functions $h_1 \in \mathcal{H}_1$, $h_2 \in \mathcal{H}_2$ such that in a single call to \textsc{Partition},
	\begin{itemize}
		\item there are no bad bins, and
		\item there are at most $\frac{\nin}{\ell^2}$ bad nodes.
	\end{itemize}
\end{lemma}

This latter property allows us to bound the size of $G_0$, the graph induced by bad nodes encountered in any recursive~call:

\begin{corollary}\label{cor:G0-is-small}
	Consider an arbitrary call of \textsc{ColorReduce}$(G,\ell)$. Then $G_0$, the graph induced by bad nodes in \textsc{Partition}$(G,\ell)$, is of size $O(\nin)$.
\end{corollary}

\begin{proof}
	By Lemma \ref{lem:deterministic-hashing}, in the call to \textsc{Partition}$(G,\ell)$, $G_0$ has at most $\frac{\nin}{\ell^2}$ nodes. Furthermore, by Corollary \ref{corollary:properties-of-Partition} (ii), the degree of each node in the call to \textsc{Partition}$(G,\ell)$ is at most $\ell + \ell^{0.7} < \ell^2$. Therefore, the graph of the bad nodes is of size at most $\frac{\nin}{\ell^2} \cdot \ell^2 = O(\nin)$, and can be collected onto a single \MPC machine and properly colored.
\end{proof}

\subsection{Recursive calls to {\sc ColorReduce} and recursion depth}
\label{sec:recursion}

Our algorithm is recursive, and it starts with the call to \textsc{ColorReduce}$(\Gin,\Delta)$ for the input graph $\Gin$ of maximum degree $\Delta$, where each node $v$ has assigned a color palette of size $p(v)$, $p(v) > d(v)$. Starting with our initial call on a graph with $n_0 = \nin$ nodes and $\ell_0 = \Delta$, we will upper bound the number of nodes and degree of graphs in subsequent recursive calls of depth~$i$. We will denote these values by $n_i$ and $\Delta_i$ respectively, and by $\ell_i$ the value of $\ell$ in a recursive call of depth~$i$.

We begin with the following bound for the values of $\ell_i$ in a recursive call of depth~$i$.

\begin{lemma}\label{lem:elli}
	$\frac12 \Delta^{0.9^i} < \ell_i \le \Delta^{0.9^i}$.
\end{lemma}

\begin{proof}
	The proof is by induction on $i$. Clearly the claim is true for $\ell_0 = \Delta$.
	
	Assuming that the claim holds for $i = j$, for $i = j + 1$ we have the following,
	\begin{align*}
	\ell_{j+1}
	&=
	\ell_j^{0.9} - \ell_j^{0.6}
	\le
	\ell_j^{0.9}
	\le
	\Delta^{0.9 \cdot 0.9^j}
	=
	\Delta^{0.9^{j+1}}
	\end{align*}
	
	and (assuming $\Delta$ is greater than a sufficiently large constant)
	\begin{align*}
	\ell_{j+1}
	&=
	\ell_j^{0.9} - \ell_j^{0.6}
	>
	\frac{1}{2^{0.9}} \Delta^{0.9 \cdot 0.9^j} - \Delta^{0.6 \cdot 0.9^j}
	>
	\frac12 \Delta^{0.9^{j+1}}
	\enspace.\qedhere
	\end{align*}
\end{proof}

Next, we bound the number of nodes $n_i$ in subsequent recursive calls of depth~$i$.

\begin{lemma}\label{lem:ni}
	$n_i \le 3^i (\nin \Delta^{0.9^i - 1} + \nin^{0.6})$.
\end{lemma}

\begin{proof}
	The proof is by induction on $i$. As a base case, $n_0 =\nin \le 3^0 (n \Delta^{0.9^{0} -1} + \nin^{0.6})=\nin+\nin^{0.6}$.
	
	Assume that the claim holds for $i = j$. Then for $i = j + 1$, by induction, by the definition of good bins (cf. Definition~\ref{def:good-and-bad}), and by Lemma \ref{lem:elli}, we obtain
	\begin{align*}
	n_{j+1}
	&\le
	2 n_j \ell_j^{-0.1} + \nin^{0.6}
	\le
	2 \left(3^j (\nin \Delta^{0.9^j - 1} + \nin^{0.6})\right) \left(2^{0.1} \Delta^{-0.1 \cdot 0.9^j}\right) + \nin^{0.6}
	\\
	&\le
	3^{j+1} \nin \Delta^{0.9^j - 1 - 0.1 \cdot 0.9^j} + 3^{j+1} \nin^{0.6}
	=
	3^{j+1} \nin \Delta^{0.9^{j+1} - 1} + 3^{j+1} \nin^{0.6}
	\enspace.\qedhere
	\end{align*}
\end{proof}

We can similarly prove a bound on degree of graphs in recursive calls of depth~$i$.

\begin{lemma}\label{lem:deltai}
	$\Delta_i \le 2^i \Delta^{0.9^i}$.
\end{lemma}

\begin{proof}
	The proof is by induction on $i$. As a base case, $\Delta_0 = \Delta = 2^0 \Delta^{0.9^0}$.
	
	Assume the claim holds for $i = j$. By the definition of good bins (cf. Definition~\ref{def:good-and-bad}), we have $\Delta_{j+1} \le \Delta_j \ell_j^{-0.1} + \ell_j^{0.6}$, and if we combine this bound with Lemma \ref{lem:elli}, by induction we obtain,
	\begin{align*}
	\Delta_{j+1}
	&\le
	\Delta_j \ell_j^{-0.1} + \ell_j^{0.6}
	\le
	\left(2^j \Delta^{0.9^j}\right) \left(2^{0.1} \Delta^{-0.1 \cdot 0.9^j}\right) + \Delta^{0.6 \cdot 0.9^j}
	<
	2^{j+1} \Delta^{0.9^{j+1}}
	\enspace.\qedhere
	\end{align*}
\end{proof}

Let us first give a simple bound for the recursion depth of our algorithm.

\begin{lemma}\label{lem:9-calls}
	After a depth of 9 recursive calls, the graph induced by each bin is of size $O(\nin)$.
\end{lemma}

\begin{proof}
	By Lemmas \ref{lem:ni} and \ref{lem:deltai}, the total size of any graph $G'$ induced by any bin after a depth of $i$ recursive calls satisfies the following,
	\begin{align*}
	|G'|
	&\le
	n_i \cdot \Delta_i
	<
	3^i (\nin \Delta^{0.9^{i} - 1} + \nin^{0.6}) \cdot 2^i \Delta^{0.9^i}
	=
	6^i (\nin \Delta^{0.9^{i} - 1} + \nin^{0.6}) \cdot \Delta^{0.9^i}
	\enspace.
	\end{align*}
	
	Plugging in $i=9$, we obtained the following bound for the total size of any graph $G'$ induced by any bin after a depth of 9 recursive calls,
	\begin{align*}
	|G'|
	&<
	6^9 (\nin \Delta^{0.9^9 - 1} + \nin^{0.6}) \cdot \Delta^{0.9^9}
	<
	6^9(\nin \Delta^{0.4 - 1} + \nin^{0.6}) \cdot \Delta^{0.4}
	\le
	6^9 (\nin \Delta^{-0.2} +\nin) = O(\nin)
	\enspace.\qedhere
	\end{align*}
\end{proof}

\subsection{Finalizing the analysis for \congc: Proof of Theorem \ref{thm:CC-coloring}}
\label{sec:final-analysis}

Now we are ready to complete our analysis and prove our main theorem, Theorem \ref{thm:CC-coloring}.

\begin{proof}[Proof of Theorem \ref{thm:CC-coloring}]
	\textsc{ColorReduce}$(\Gin,\Delta)$ creates a recursion tree of depth at most $9$, resulting in $O(1)$ groups of instances of size $O(\nin)$ to be collected onto machines and colored sequentially. All operations within each recursive call (including coloring graphs $G_0$, see Corollary \ref{cor:G0-is-small}) can be implemented in a constant number of rounds in \congc and on \MPC with $O(\nin)$ local space on any machine (and $O(\nin\Delta)$ total space), so the total number of rounds required is also constant.
\end{proof}

\subsection{Note on optimal global space for ($\Delta+1$)-coloring (Theorem \ref{thm:MPC-coloring-congc})}
\label{subsec:reducing-space-for-Delta+1-coloring-on-MPC}

For the special case of ($\Delta+1$)-coloring, where initial palettes are all $[\Delta+1]$ and therefore need not be given as input, we remark that explicitly maintaining palettes for all nodes costs $\Theta(\nin\Delta)$ global space, which can exceed the optimal global space bound of $O(\nim+\nin)$. We address this issue as follows:

We require palettes in two places: when applying \textsc{Partition}, and when coloring instances of size $O(\nin)$ locally. In this latter case, we can afford to arbitrarily drop colors from a node $v$'s palette until its palette size is $d(v)+1$, so the space used in total is $O(\nim)$.

In the former case, consider applying \textsc{Partition} to an instance with $n$ nodes and $m$ edges. Nodes will be distributed over $O(1+\frac{m}{\nin})$ machines in order to apply the method of conditional expectations, for which they require palette access.

Notice that during the run of our algorithm, node palettes are updated in two ways: by restriction based on partition by some hash function, and by removing colors when they are used to color neighbors. We will maintain the palette of each node $v$ by explicitly storing colors used by its neighbors, and implicitly storing colors removed during partitioning by storing the chosen hash function. In this way, nodes' palettes are fully specified and can be queried. For any node $v$, we explicitly store up to one color per neighbor, which requires $d(v)$ space for $v$ and therefore total $O(\nim)$ for all nodes.

So, it remains to analyze the space cost of storing the necessary hash functions. In our \textsc{Partition} call on $n$ nodes and $m$ edges, we have $O(1+\frac{m}{\nin})$ machines which must each store the $O(1)$ hash functions which have previously been applied to this instance. These hash functions are specified by $O(\log \nin)$ bits, so the total space required is $O(\log \nin+\frac{m\log \nin}{\nin})$.

We may run up to $O(\sqrt\nin)$ calls to \textsc{Partition} concurrently (we cannot have more, because we must have $m=\Omega(\nin)$, so $n=\Omega(\sqrt\nin)$). Furthermore, each edge in $\Gin$ is present in at most one call. So, the total space used to store hash functions over \emph{all} concurrent \textsc{Partition} calls is $O(\sqrt\nin\log \nin+\frac{\nim\log \nin}{\nin}) = O(\nin+\nim)$, which yields Theorem \ref{thm:MPC-coloring-congc}.



\section{Coloring in Low-Space \MPC}
We now extend our methods to the more restrictive regime of low-space \MPC. Again, our algorithm will be a recursive procedure which relies on a derandomized partitioning of nodes and colors into smaller instances. However, reducing to instances of size $O(\nin)$ is no longer sufficient to allow collection onto single machines for a constant-round solution. So, we instead reduce until degree is $\nin^{7\delta}$ for some sufficiently small constant $\delta$, and then we can apply a reduction to maximal independent set (MIS) in order to use an existing algorithm of \cite{CDP19}, requiring $O(\log \Delta+\log\log \nin)$ rounds. This is the reason for the difference in round complexity compared to Algorithms \ref{alg:ColorReduce}, \ref{alg:Partition}.

\subsection{Reduction to MIS in low degree instances}

When degree is $O(n^{7\delta})$, we apply the standard (due to Luby \cite{Luby86}) reduction to maximal independent set (MIS): a new graph is created in which nodes $v$ in the original graph correspond to cliques of size $p(v)$ in the new graph, with each clique node representing a color from the palette of node $v$. If two neighboring original nodes share a color between their palettes, an edge is drawn between their corresponding clique nodes. An MIS in this new graph necessarily indicates a complete proper coloring of the original graph, since it can easily be seen that exactly one clique node from each clique joins the MIS, and this determines the color for the original node. When the reduction is applied to a graph with $\hat\nin$ vertices and maximum degree $\nin^{7\delta}$, the new reduction graph has at most $O(\hat\nin \cdot \nin^{7\delta})$ vertices, and maximum degree at most $\nin^{14\delta}$.

Then, we can solve $(\deg+1)$-list coloring by applying the MIS algorithm of \cite{CDP19} to this reduction graph. This algorithm requires $O(\log \Delta+\log\log \nin)$ rounds, and can be applied to our reduction graphs using $O(\nin^{\delta})$ and $O(\hat\nin^{1+\delta}\nin^{21\delta})$ local and global space respectively.

\subsection{Main low-space \MPC algorithm}

We now present our main algorithm for coloring in low-space \MPC, which recursively reduces instances until they have degree $O(\nin^{7\delta})$ and then applies the MIS reduction.

\begin{algorithm}[H]
	\caption{\textsc{LowSpaceColorReduce}$(G)$}
	\label{alg:LSColorReduce}
	\begin{algorithmic}
		\State $G_0, \dots, G_{\nin^{\delta}} \gets \textsc{LowSpacePartition}(G)$.
		\State For each $i = 1, \dots,\nin^{\delta}-1$, perform \textsc{LowSpaceColorReduce}$(G_i)$ in parallel.
		\State Update color palettes of $G_{\nin^{\delta}}$, perform \textsc{LowSpaceColorReduce}$(G_{\nin^{\delta}})$.
		\State Update color palettes of $G_0$, color $G_0$ using MIS reduction.
	\end{algorithmic}
\end{algorithm}

Again, we employ a partitioning procedure to divide the input instance into \emph{bins}, which are then solved recursively:

\begin{algorithm}[H]
	\caption{\textsc{LowSpacePartition}$(G)$}
	\label{alg:LSPartition}
	\begin{algorithmic}
		\State Let $G_0$ be the graph induced by the set $V_0$ of nodes $v$ with $d(v)\le \nin^{7\delta}$.
		\State Let hash function $h_1:[\nin]\rightarrow [\nin^{\delta}]$ map each node $v\notin V_0$ to a bin $h_1(v) \in [\nin^{\delta}]$.
		\State Let hash function $h_2:[\nin^2]\rightarrow [\nin^{\delta}-1]$ map colors $\gamma$ to a bin $h_2(\gamma) \in [\nin^{\delta}-1]$.
		\State Let $G_1,\dots,G_{\nin^{\delta}}$ be the graphs induced by bins $1,\dots,\nin^{\delta}$ respectively.
		\State Restrict palettes of nodes in $G_1,\dots,G_{\nin^{\delta}-1}$ to colors assigned by $h_2$ to corresponding bins.
		\State Return $G_0,\dots,G_{\nin^{\delta}}$.
	\end{algorithmic}
\end{algorithm}

The structure of this algorithm is very similar to Algorithms \ref{alg:ColorReduce} and \ref{alg:Partition}. The main structural difference is that $G_0$ is no longer a graph of \emph{bad} nodes (we will show that we can choose hash functions so that, in essence, no nodes are bad), but instead a graph of the nodes with low enough degree to apply the MIS reduction. By dealing with low-degree nodes throughout the algorithm in this way, we can account for differences in palette sizes, and thereby solve $(\deg+1)$-list coloring instead of just $(\Delta+1)$-list coloring.

The only invariant we maintain during the course of \textsc{LowSpaceColorReduce} is that all nodes have sufficient colors to color themselves, i.e. $p(v)>d(v)$. This is clearly true for our initial call on input graph $\Gin$, and later (in Lemma \ref{lem:LSPaletteSize}) it is preserved by \textsc{LowSpaceColorPartition}.

We would like to (analogously to Definition \ref{def:good-and-bad}) define \emph{good} nodes whose behavior does not differ too much from what we would expect from a random partition (and \emph{bad} nodes for whom the opposite is true). However, one of the restrictions of the low-space regime is that we cannot store a node's palette, or all of its adjacent edges, on a single machine, and therefore machines could not determine locally whether a node is good or bad. So, we will instead divide nodes' neighbors and palettes into smaller sets which \emph{do} fit onto single machines, and define a notion of a \emph{machine} being good.

Specifically, for each node $v\notin G_0$, we will create a set $M_N^v$ of machines which will be responsible for $v$'s neighbors. We partition $v$'s neighbor set among machines in $M_N^v$, with each machine receiving between $\nin^{7\delta}$ and $2\nin^{7\delta}$ neighbors.
Similarly, for each node $v\notin G_0 \cup G_{n^{\delta}}$, we create a set $M_C^v$ of machines responsible for the colors in its palette, and partition the palette among the machines, with each machine receiving between $\nin^{7\delta}$ and $2\nin^{7\delta}$ colors.
In both cases, this is possible since for $v\notin G_0$, $p(v)>d(v)>\nin^{7\delta}$.

For consistency of notation, for a machine $x\in M_N^v$, we define $d(x)$ to be the number of neighbors it receives, and $d'(x)$ to be the number of such neighbors which are hashed to the same bin as $v$.
Similarly, for a machine $x\in M_C^v$ holding colors for a node $v$, we define $p(x)$ to be the number of colors received, and $p'(x)$ to be the number of such colors hashed by $h_2$ to the same bin as $v$.

\begin{definition}[\textbf{Good and bad machines}]\mbox{}\label{def:LSgood-and-bad}
	\begin{itemize}
		\item A machine $x\in M_N^v$ is \emph{good} if $|d'(x) - d(x) \nin^{-\delta}| \le d(x)^{0.6}$.
		\item A machine $x\in M_C^v$ is \emph{good} if $p'(x) > p(x)  \nin^{-\delta} + p(x)^{0.7}$.
		\item Machines are \emph{bad} if they are not \emph{good}.
	\end{itemize}
\end{definition}

Our next step is to show that, under a \emph{random} choice of hash functions during partitioning, with high probability there are no bad machines.

\subsection{Good and bad machines for \emph{random} hash functions}
\label{subsec:LS-bad-and-good-for-random-hash-function}

We apply the bounded-independence concentration bound to show that node degrees and palette sizes within bins to not differ too much from their expectation:

\begin{lemma}\label{lem:LScoloreduce1}
	For each machine $x \in M_N^v$, $|d'(x) - d(x) \nin^{-\delta}| \le  d(x)^{0.6}$ with probability at least $1-\nin^{-2}$.
\end{lemma}

\begin{proof}
	We apply Lemma \ref{lem:conc} with $Z_1, \dots, Z_{d(x)}$ as indicator random variables for the events that each neighbor held is hashed to $v$'s bin. These variables are $(c-1)$-wise independent and each have expectation $\nin^{-\delta}$. Since we assume $c$ to be a sufficiently high constant and we have $d(x)\ge \nin^{7\delta}$, by Lemma \ref{lem:conc}, we obtain $\Prob{|Z-\Exp{Z}| \ge d(x)^{0.6}} \le d(x)^{-0.2c}\le \nin^{-2}$.	Since $\Exp{Z} = d(x)\nin^{-\delta}$, we have the required claim.
\end{proof}

\begin{lemma}\label{lem:LScoloreduce2}
	For each machine $x \in M_C^v$, $p'(v) > p(v)  \nin^{-\delta} + p(v)^{0.7}$ with probability at least $1-\nin^{-2}$.
\end{lemma}

\begin{proof}
	We apply Lemma \ref{lem:conc}, with $Z_1,\dots,Z_{p(x)}$ as indicator random variables for the events that each color held is hashed to $v$'s bin. These variables are $(c-1)$-wise independent and each have expectation $\frac{1}{\nin^{\delta}-1}$. Since we assume $c$ to be a sufficiently high constant and have $p(x) \ge \nin^{7\delta}$, by Lemma \ref{lem:conc}, we obtain $\Prob{|Z-\Exp{Z}| \ge p(x)^{0.6}} \le p(x)^{-0.2c}\le \nin^{-2}$.
	
	Therefore, since $Z = p'(x)$ and $\Exp{Z} = \frac{p(x)}{\nin^{\delta}-1}$, the probability that $p'(x) \le \frac{p(x)}{\nin^{\delta}-1} - p(x)^{0.6}$ is at most $p\nin^{-2}$.
	
	Furthermore, since $p(v) \ge \nin^{7\delta}$,
	\begin{align*}
	\frac{p(v)}{\nin^{\delta}-1} - \frac{p(v)}{\nin^{\delta}}
	&=
	\frac{p(v)}{\nin^{\delta}(\nin^{\delta}-1)}
	\ge
	\frac{p(v)}{\nin^{2\delta}}
	>
	\frac{2p(v)}{p(v)^{0.3}}
	=
	2p(v)^{0.7}
	\enspace.
	\end{align*}
	Hence, with probability at least $1 - \nin^{-2}$ we obtain the following,
	\begin{align*}
	p'(v)
	&>
	\frac{p(v)}{\nin^{\delta}-1} - p(v)^{0.6}
	>
	\frac{p(v)}{\nin^{\delta}} + p(v)^{0.7}
	\enspace.\qedhere
	\end{align*}
\end{proof}

We will now define a cost function for pairs of hash functions that we would like to minimize, analogously to our definition in Section \ref{subsec:bad-and-good-for-random-hash-function}. This time our cost function is simpler, since we can weight all of our bad events equally, and with high probability none of them occur.

We define the \emph{cost} function $\qual(h_1,h_2)$ of a pair of hash functions $h_1 \in \mathcal{H}_1$ and $h_2 \in \mathcal{H}_2$, as follows:
\begin{align}
\label{def:LSquality}
\qual(h_1,h_2)
&=
|\{\text{bad machines under }h_1,h_2\}|
\enspace.
\end{align}

We can then bound the cost of random hash function pairs:

\begin{lemma}\label{lemma:LSexpected-quality}
	The expected cost of a random hash function pair from $\mathcal{H}_1 \times \mathcal{H}_2$ is less than $1$.
\end{lemma}

\begin{proof}
	By Lemmas \ref{lem:LScoloreduce1} and \ref{lem:LScoloreduce2}, any machine is bad with probability at most $\nin^{-2}$. The number of machines we require is $O(\nin+ \frac{\nim}{\nin^{7\delta}})=o(\nin^2)$. Therefore,
	\begin{align*}
	\Exp{\qual(h_1,h_2)}
	&=
	\Exp{|\{\text{bad machines}\}|}
	\le
	\nin^{-2} \cdot o(\nin^2)
	<
	1
	\enspace.\qedhere
	\end{align*}
\end{proof}

Applying the method of conditional expectations therefore gives the following:

\begin{lemma}\label{lem:LSdeterministic-hashing}
	In $O(1)$ \MPC rounds with $O(\nin^{7\delta})$ local space per  machine and $O(\nim+\nin)$ total space, one can select hash functions $h_1 \in \mathcal{H}_1$, $h_2 \in \mathcal{H}_2$ such that in a single call to \textsc{LowSpacePartition},
	\begin{itemize}
		\item for any node $v\notin G_0$, $d'(v) < 2d(v)\nin^{-\delta}$, and
		\item for any node $v\notin G_0\cup G_{\nin^{\delta}}$, $d'(v)<p'(v)$.
	\end{itemize}
\end{lemma}

\begin{proof}
	By the method of conditional expectations (using, as local functions $q_x$, the indicator variables for the event that machine $x$ is bad), we can select a hash function pair $h_1,h_2$ with cost $q(h_1,h_2)<1$; since cost must be a integer, this means that $q(h_1,h_2)=0$, i.e. all machines are good.
	
	By Definition \ref{def:LSgood-and-bad}, for any node $v\notin G_0$, we therefore have:
	\begin{align*}
	d'(v)
	&=
	\sum_{x \in M_N^v} d'(x)
	\le
	\sum_{x \in M_N^v} \left(d(x)\nin^{-\delta}+d(x)^{0.6} \right)
	\le
	d(v)\nin^{-\delta} + |M_N^v|(2\nin^{7\delta})^{0.6}
	<
	d(v)\nin^{-\delta} + \nin^{4.5\delta}\frac{d(v)}{\nin^{7\delta}}
	\\&=
	d(v)\nin^{-\delta}+d(v)\nin^{-2.5\delta}
	\enspace.
	\end{align*}
	
	Since $d(v)\nin^{-\delta}+d(v)\nin^{-2.5\delta}<2d(v)\nin^{-\delta}$, this satisfies the first part of the lemma.
	
	For any node $v\notin G_0\cup G_{\nin^{\delta}}$, we also have:
	\begin{align*}
	p'(v)
	&=
	\sum_{x \in M_C^v} p'(x)
	>
	\sum_{x \in M_C^v} \left(p(x)\nin^{-\delta}+p(x)^{0.7} \right)
	\ge
	p(v)\nin^{-\delta} + |M_C^v|(\nin^{7\delta})^{0.7}
	\\&\ge
	p(v)\nin^{-\delta} + \nin^{4.9\delta}\frac{p(v)}{2\nin^{7\delta}}
	>
	p(v)\nin^{-\delta}+p(v)\nin^{-2.5\delta}
	>
	d(v)\nin^{-\delta}+(v)\nin^{-2.5\delta}
	\ge
	d'(v)
	\enspace.\qedhere
	\end{align*}
\end{proof}

We can therefore show that whenever we call \textsc{LowSpacePartition}, we maintain our invariant that all nodes have sufficient colors, i.e. $p'(v)>d'(v)$.

\begin{lemma}\label{lem:LSPaletteSize}
	Assume that at the start of our \textsc{LowSpacePartition}$(G)$ call, all nodes $v$ in $G$ have $p(v)>d(v)$. Then, after the call, we have $p'(v)>d'(v)$.
\end{lemma}

\begin{proof}
	By Lemma \ref{lem:LSdeterministic-hashing}, the claim holds for all $v\notin G_0\cup G_{\nin^{\delta}}$. For a node $v \in G_0\cup G_{\nin^{\delta}}$, $v$'s palette is not restricted by partitioning and colors are only subsequently removed when used by $v$'s neighbors, so it follows that $p'(v)>d'(v)$.
\end{proof}

Now we are ready to prove Theorem \ref{thm:MPC-coloring}, which states that Algorithm \ref{alg:LSColorReduce} applied to the input graph \Gin\ performs $(\deg+1)$-list coloring in $O(\log \Delta + \log\log\nin)$ rounds of \MPC, with $O(\nin^\eps)$ space per machine and $O(\nim+\nin^{1+\eps})$ total space.

\begin{proof}[Proof of Theorem \ref{thm:MPC-coloring}]
	Each call of \textsc{LowSpaceColorReduce} makes two sequential sets of parallel recursive calls, on instances with maximum degree a factor of $\frac12 \nin^{\delta}$ lower than input maximum degree, by Lemma \ref{lem:LSdeterministic-hashing}. (It also makes a call to the MIS reduction taking $O(\log \Delta + \log\log\nin)$ rounds.) So, after a recursion depth of $\log_{\frac12 \nin^{\delta}} \Delta = O(1)$, the maximum degree is at most $O(\nin^{7\delta})$, and no further recursive calls are made. Since the depth of the recursion tree is $O(1)$, so is the number of \textsc{LowSpaceColorReduce} calls performed sequentially. So, the total round complexity is asymptotically equivalent to the complexity of a single \textsc{LowSpaceColorReduce} call, which is dominated by the $O(\log \Delta + \log\log\nin)$ rounds of the MIS reduction.
	
	Calls to \textsc{LowSpacePartition} use $O(\nin^{7\delta})$ local space, and all concurrent calls use $O(\nin+\nim)$ global space in total, since we never store more than $O(1)$ copies of any node, edge or palette entry.
	
	The global space complexity is dominated by the calls to MIS reduction, which as mentioned requires $O(\nin^\delta)$ and $O(\hat\nin^{1+\delta}\nin^{21\delta})$ local and global space respectively, where $\hat\nin$ is the number of nodes in the graph on which the call is made. We make many concurrent calls to the MIS reduction, but each node of \Gin~ is part of only one MIS reduction call. Therefore, the total space requirements for all concurrent calls are at most $O(\nin^{\delta})$ local space and $O(\nin^{1+22\delta})$ global space. Setting $\delta\le\frac{\epsilon}{22}$, this is $O(\nin^{\eps})$ and $O(\nin^{1+\eps})$ respectively.
\end{proof}

\section{Conclusions}
In this paper we have presented a constant-round deterministic $(\Delta+1)$-list coloring algorithm for the \congc and linear-space \MPC models, which not only greatly improves the deterministic complexity of the problem to optimality, but is also significantly simpler than the recent constant-round \emph{randomized} algorithm \cite{chang2019complexity}. Our algorithm relies on the derandomization of a recursive graph and palette sparsification procedure. We also extended this approach to low-space \MPC; here the $O(\log\Delta + \log\log n)$ round complexity arises from the fact that we cannot collect instances onto single machines, and instead reduce to the problem to MIS and apply an existing algorithm once degree is sufficiently low.

With this work we settle the deterministic complexity of coloring in \congc; the complexity in low-space \MPC, however, remains open. A corresponding constant-round algorithm there seems unlikely, due to the $\Omega(\log\log\log\nin)$ conditional lower bound \cite{GKU19}, but there is no reason to believe one could not improve the dependency on $\Delta$ to sub-logarithmic. Another area for possible improvement is our $O(\nim+\nin^{1+\eps})$ global space bound, which is slightly weaker than the optimal $O(\nim+\nin)$.


\newcommand{\Proc}{Proceedings of the~}
\newcommand{\DISC}{International Symposium on Distributed Computing (DISC)}
\newcommand{\FOCS}{IEEE Symposium on Foundations of Computer Science (FOCS)}
\newcommand{\ICALP}{Annual International Colloquium on Automata, Languages and Programming (ICALP)}
\newcommand{\IPCO}{International Integer Programming and Combinatorial Optimization Conference (IPCO)}
\newcommand{\ISAAC}{International Symposium on Algorithms and Computation (ISAAC)}
\newcommand{\JACM}{Journal of the ACM}
\newcommand{\NIPS}{Conference on Neural Information Processing Systems (NeurIPS)}
\newcommand{\OSDI}{Conference on Symposium on Opearting Systems Design \& Implementation (OSDI)}
\newcommand{\PODS}{ACM SIGMOD Symposium on Principles of Database Systems (PODS)}
\newcommand{\PODC}{ACM Symposium on Principles of Distributed Computing (PODC)}
\newcommand{\SICOMP}{SIAM Journal on Computing}
\newcommand{\SIROCCO}{International Colloquium on Structural Information and Communication Complexity}
\newcommand{\SODA}{Annual ACM-SIAM Symposium on Discrete Algorithms (SODA)}
\newcommand{\SPAA}{Annual ACM Symposium on Parallel Algorithms and Architectures (SPAA)}
\newcommand{\STACS}{Annual Symposium on Theoretical Aspects of Computer Science (STACS)}
\newcommand{\STOC}{Annual ACM Symposium on Theory of Computing (STOC)}


\bibliographystyle{plain}
\bibliography{ColoringArxiv}

\begin{thebibliography}{10}

\bibitem{BKM20}
Philipp Bamberger, Fabian Kuhn, and Yannic Maus.
\newblock Efficient deterministic distributed coloring with small bandwidth.
\newblock In {\em \Proc 39th \PODC}, pages 243--252, 2020.

\bibitem{barenboim2018distributed}
Leonid Barenboim and Victor Khazanov.
\newblock Distributed symmetry-breaking algorithms for {Congested} {Cliques}.
\newblock In {\em \Proc 13th International Computer Science Symposium in Russia
  (CSR)}, pages 41--52, 2018.

\bibitem{BDH18}
Soheil Behnezhad, Mahsa Derakhshan, and MohammadTaghi Hajiaghayi.
\newblock Brief announcement: {Semi-MapReduce} meets {Congested Clique}.
\newblock {\em CoRR abs/1802.10297}, 2018.

\bibitem{BR94}
Mihir Bellare and John Rompel.
\newblock Randomness-efficient oblivious sampling.
\newblock In {\em \Proc 35th \FOCS}, pages 276--287, 1994.

\bibitem{CPS17}
Keren Censor{-}Hillel, Merav Parter, and Gregory Schwartzman.
\newblock Derandomizing local distributed algorithms under bandwidth
  restrictions.
\newblock In {\em \Proc 31st \DISC}, pages 11:1--11:16, 2017.

\bibitem{chang2019complexity}
Yi-Jun Chang, Manuela Fischer, Mohsen Ghaffari, Jara Uitto, and Yufan Zheng.
\newblock The complexity of {$(\Delta+1)$} coloring in {Congested Clique},
  massively parallel computation, and centralized local computation.
\newblock In {\em \Proc 38th \PODC}, pages 471--480, 2019.

\bibitem{CDP19}
Artur Czumaj, Peter Davies, and Merav Parter.
\newblock Graph sparsification for derandomizing massively parallel computation
  with low space.
\newblock In {\em \Proc 32nd \SPAA}, pages 175--185, 2020.

\bibitem{DG04}
Jeffrey Dean and Sanjay Ghemawat.
\newblock {MapReduce}: Simplified data processing on large clusters.
\newblock In {\em \Proc 6th \OSDI}, pages 10--10, 2004.

\bibitem{DG08}
Jeffrey Dean and Sanjay Ghemawat.
\newblock {MapReduce}: Simplified data processing on large clusters.
\newblock {\em Commununication of the ACM}, 51(1):107--113, January 2008.

\bibitem{GKU19}
Mohsen Ghaffari, Fabian Kuhn, and Jara Uitto.
\newblock Conditional hardness results for massively parallel computation from
  distributed lower bounds.
\newblock In {\em \Proc 60th \FOCS}, pages 1650--1663, 2019.

\bibitem{GSZ11}
Michael~T. Goodrich, Nodari Sitchinava, and Qin Zhang.
\newblock Sorting, searching, and simulation in the {MapReduce} framework.
\newblock In {\em \Proc 22nd \ISAAC}, pages 374--383, 2011.

\bibitem{Han96}
Yijie Han.
\newblock A fast derandomization scheme and its applications.
\newblock {\em \SICOMP}, 25(1):52--82, 1996.

\bibitem{IBYBF07}
Michael Isard, Mihai Budiu, Yuan Yu, Andrew Birrell, and Dennis Fetterly.
\newblock Dryad: Distributed data-parallel programs from sequential building
  blocks.
\newblock {\em SIGOPS Operating Systems Review}, 41(3):59--72, March 2007.

\bibitem{KSV10}
Howard~J. Karloff, Siddharth Suri, and Sergei Vassilvitskii.
\newblock A model of computation for {MapReduce}.
\newblock In {\em \Proc 21st \SODA}, pages 938--948, 2010.

\bibitem{Lenzen13}
Christoph Lenzen.
\newblock Optimal deterministic routing and sorting on the {Congested Clique}.
\newblock In {\em \Proc 32nd \PODC}, pages 42--50, 2013.

\bibitem{LPP15}
Zvi Lotker, Boaz {Patt-Shamir}, and Seth Pettie.
\newblock Improved distributed approximate matching.
\newblock {\em Journal of the ACM}, 62(5):38:1--38:17, November 2015.

\bibitem{Luby86}
Michael Luby.
\newblock A simple parallel algorithm for the maximal independent set problem.
\newblock {\em {SIAM} Journal on Computing}, 15(4):1036--1053, 1986.

\bibitem{Parter18}
Merav Parter.
\newblock ({$\Delta+1$}) coloring in the {Congested Clique} model.
\newblock In {\em \Proc 45th \ICALP}, pages 160:1--160:14, 2018.

\bibitem{ParterSu18}
Merav Parter and Hsin{-}Hao Su.
\newblock Randomized {$(\Delta+1)$}-coloring in {$O(\log^*\Delta)$} {Congested}
  {Clique} rounds.
\newblock In {\em \Proc 32nd \DISC}, pages 39:1--39:18, 2018.

\bibitem{rozhon2019polylogarithmic}
V{\'a}clav Rozho{\v{n}} and Mohsen Ghaffari.
\newblock Polylogarithmic-time deterministic network decomposition and
  distributed derandomization.
\newblock In {\em \Proc 52nd \STOC}, pages 350--363, 2020.

\bibitem{Vadhan12}
Salil~P. Vadhan.
\newblock Pseudorandomness.
\newblock {\em Foundations and Trends in Theoretical Computer Science},
  7(1-3):1--336, 2012.

\bibitem{White12}
Tom White.
\newblock {\em Hadoop: The Definitive Guide}.
\newblock O'Reilly Media, Inc., 2012.

\bibitem{ZCFSS10}
Matei Zaharia, Mosharaf Chowdhury, Michael~J. Franklin, Scott Shenker, and Ion
  Stoica.
\newblock Spark: Cluster computing with working sets.
\newblock In {\em \Proc 2nd {USENIX} Workshop on Hot Topics in Cloud Computing
  (HotCloud)}, 2010.

\end{thebibliography}


\end{document}